 \documentclass[conference]{IEEEtran}

\IEEEoverridecommandlockouts

\usepackage[utf8]{inputenc}
\usepackage{amsmath,amsthm,relsize}
\usepackage{amsfonts, amssymb, cuted}
\usepackage{hyperref}
\usepackage{cleveref}
\usepackage{mathtools}
\usepackage{algorithm}
\usepackage[noend]{algpseudocode}
\usepackage{cite}
\usepackage[table]{xcolor}
\usepackage{soul}
\usepackage{graphicx}
\usepackage{pgfplotstable}
\usepackage{pgfplots}
\usepackage{nicefrac}
\usepackage{enumitem}
\usepackage{support-caption}
\usepackage{subcaption}
\usepackage[font=footnotesize]{caption}
\usepackage{multirow}

\pgfplotsset{compat=1.14}

\newcommand{\subparagraph}{}
\usepackage{titlesec}
\titlespacing\section{3pt}{6pt plus 4pt minus 2pt}{6pt plus 2pt minus 2pt}
\titlespacing\subsection{3pt}{4pt plus 4pt minus 2pt}{4pt plus 2pt minus 2pt}
\titlespacing\subsubsection{3pt}{3pt plus 4pt minus 2pt}{0pt plus 2pt minus 3pt}

\newtheorem{thm}{Theorem}
\newtheorem{lem}{Lemma}

\newtheorem{cor}{Corollary}

\newtheorem{exmp}{Example}
\theoremstyle{definition}

\newcommand{\CK}[0]{{\mathcal{K}}}

\newcommand{\CS}[0]{{\mathcal{S}}}
\newcommand{\CT}[0]{{\mathcal{T}}}
\newcommand{\CU}[0]{{\mathcal{U}}}
\newcommand{\CV}[0]{{\mathcal{V}}}

\begin{document}

\title{Non-Symmetric Coded Caching\\for Location-Dependent Content Delivery
\vspace{-5 mm}
\thanks{This work was supported by the Academy of Finland under grants no. 319059 (Coded Collaborative Caching for Wireless Energy Efficiency) and 318927 (6Genesis Flagship).}}
\date{Jan 2021}

\author{
\IEEEauthorblockN{
Hamidreza Bakhshzad Mahmoodi, MohammadJavad Salehi and Antti T\"olli
}
\\
\vspace{-5 mm}
\IEEEauthorblockA{
Center for Wireless Communications, University of Oulu, 90570 Finland}
\IEEEauthorblockA{
Email: \{fist\_name.last\_name\}@oulu.fi}
}

\maketitle

\begin{abstract}
Immersive viewing is emerging as the next interface evolution for human-computer interaction. A truly wireless immersive application necessitates immense data delivery with ultra-low latency, raising stringent requirements for next-generation wireless networks.
A potential solution for addressing these requirements is through the efficient usage of in-device storage and computation capabilities. This paper proposes a novel location-based coded cache placement and delivery scheme, which leverages the nested code modulation (NCM) to enable multi-rate multicasting transmission. To provide a uniform quality of experience in different network locations, we formulate a linear programming cache allocation problem. Next, based on the users' spatial realizations, we adopt an NCM based coded delivery algorithm to efficiently serve a distinct group of users during each transmission. Numerical results demonstrate that the proposed location-based delivery method significantly increases transmission efficiency compared to state of the art.
\end{abstract}

\begin{IEEEkeywords}
Coded Caching; Location-Dependent Caching; Immersive Viewing
\end{IEEEkeywords}

\section{Introduction}
The current dominant interface for human-computer interaction is through mobile flat-screen devices such as smartphones and tablets. The next interface evolution is expected to bring forward wireless immersive viewing experiences facilitated by more capable wearable gadgets submerging users into the three-dimensional~(3D) digital world. However, such an evolution requires powerful and agile external radio connections. It imposes extremely stringent key performance indicators (KPIs) that are simply beyond what is possible with the current networking standards~\cite{han2017problem}. This necessitates new techniques that leverage recent advances in communication, storage, computing, big data analysis and machine learning~\cite{bastug2017toward}.

In this regard, using cheap on-device storage for improving bandwidth efficiency is considered a promising technique~\cite{bastug2014living}. In~\cite{sun2019communications}, it is shown that utilizing caching and computing capabilities of mobile VR devices effectively alleviates the traffic burden over the wireless network.
In~\cite{maddah2014fundamental}, a new caching technique known as \emph{Coded Caching (CC)} is introduced. In CC, well-defined fragments of all the files in the library are stored in the cache memories, creating a cache-aided multicasting gain during the delivery phase and resulting in a \emph{global caching gain}.
%
%
The original CC scheme in~\cite{maddah2014fundamental} is extended to multi-server networks in~\cite{shariatpanahi2016multi}, and later to wireless multi-antenna systems in~\cite{shariatpanahi2018multi,tolli2017multi}. A device-to-device (D2D) extension of multi-antenna coded caching is also investigated in~\cite{mahmoodi2020d2d}. Meantime, various practical limitations of coded caching have been addressed by the research community. Most notably, the large subpacketization requirement, defined as the number of smaller parts each file should be split into, is addressed in~\cite{lampiris2018adding,salehi2020lowcomplexity}, while the effect of the subpacketization on the low-SNR rate is investigated in~\cite{salehi2019subpacketization}.

A less-studied problem of coded caching schemes, affecting content delivery applications in general and immersive viewing applications in particular, is the \textit{near-far} issue. Due to the underlying multicasting nature of coded caching schemes, the achievable rate in any multicast message is limited to the rate of the user with the worst channel conditions. In~\cite{destounis2020adaptive}, a congestion control technique is proposed to avoid serving users in adverse fading conditions, while in~\cite{salehi2020coded} multiple descriptor codes~(MDC) are utilized to serve ill-conditioned users with a lower quality-of-experience. Unlike~\cite{destounis2020adaptive,salehi2020coded}, which are based on traditional XOR-ing of data elements, in~\cite{tang2017coded} nested code modulation (NCM) is proposed to allow building codewords that serve every user in the multicasting group with a different rate. This multi-rate property is achieved by altering the modulation constellation using side information on the file library and other users' requests.
A similar multi-rate transmission scheme can also be found in~\cite{chen2010novel}. However, the aforementioned approaches are not suitable for dynamic real-time applications, where users frequently move inside the network, and their achievable rate changes accordingly.

This paper introduces a new location-dependent coded caching scheme for efficient content delivery in wireless access networks. We consider a wireless communication scenario in which the users are free to move, and their requested content depends on their current location. Furthermore, the requested content at each location is assumed to be of the same size. As a specific use-case, we assume a multi-user immersive viewing environment where a group of users is submerged into a network-based immersive application that runs on a high-end eye-wear. Such a use-case necessitates heavy multimedia traffic and guaranteed user quality of experience (QoE) throughout the operating environment. In this regard, a location-dependent, uneven memory allocation is carried out based on the attainable data rate at each given location. Moreover, a novel multicast transmission scheme with an underlying NCM structure to support different data rates is devised to deliver the missing user-specific content within the same multicast transmission. Due to the optimized location-dependent cache placement, the worst-case delivery time is minimized across all the locations.

\section{System Model}
\label{section:system_model}
We envision a bounded environment (game hall, operating theatre, etc.) in which a single-antenna server serves $K$ single-antenna users through wireless communication links. The set of users is denoted by $\CK = \{ 1, ..., K \}$. The users are equipped with finite-size cache memories and are free to move throughout the environment. At each time slot, every user requests data from the server based on the application needs and its location. The requested data content can be divided into static and dynamic parts where the former can be proactively stored in the user cache memories. This paper focuses on the wireless delivery of the static location-dependent content partially aided by in-device cache memories.\footnote{We assume that a portion of the achievable data rate available at each user is dedicated to deliver the dynamic content without cache assistance.} A real-world application of this communication setup is a wireless immersive digital experience environment where the requested data is needed for reconstructing the location-dependent 3D field-of-view (FoV) at each user. The goal is to design a cache-aided communication scheme that minimizes the maximum required delivery time to transmit all the requested data to the served users. In other words, the aim is to provide a uniform QoE, irrespective of the users' location.

Intuitively, a larger share of the total cache memory should be reserved for storing data needed in locations where the communication quality is poor. We split the environment into $S$ regions, such that all points in a given region have almost the same distance from the server (i.e., can be served approximately with the same data rate). In the following, we refer to these regions as states and denote the set of states as $\CS$. A graphical example of an application environment with its states is provided in Figure~\ref{fig:system-model}. The file required for reconstructing the FoV of state $j \in \CS$ is denoted by $W(j)$. We assume for every region $j\in\CS$, the size of $W(j)$ is $F$ bits, and every user is equipped with a cache memory of size $MF$ bits. For the sake of simplicity, we consider a normalized data unit and drop $F$ in subsequent notations. Moreover, we consider the delivery procedure in a specific time slot and ignore the time index (the same procedure is repeated every time slot).

We assume a wideband communication scheme, where the total bandwidth is divided into several small frequency bins. 
For simplicity, we assume that the attainable expected rate at a particular location roughly depends on the transmitted power and the distance between the location and the server. Thus, the expected data rate attained in state $j\in \CS$ is expressed as
    $\bar{r}(j) = \log(1+\frac{P d^{-n}(j)}{N_0})$,
where $P$ is the transmission power, $N_0$ is the additive white Gaussian noise power, $n$ is the path-loss component, and $d(j)$ is the maximum distance from the server to all the points in state $j$.\footnote{
We need an estimation of the achievable rate at different states to perform the location-dependent cache placement. However, during the delivery phase, the communication can be performed with real achievable rates calculated using the available channel state information (CSI). In this paper, for notational and analysis simplicity, we have assumed the proposed rate estimation in the placement phase, based on the simple path loss model and line of sight (LOS) communications, is also valid throughout the delivery phase. Relaxing this assumption is left for the extended version of the paper.
} Thus, the expected data rate over all the frequency bins is approximated as $\hat{r}(j) \sim B\bar{r}(j)$, where $B$ is the communication bandwidth. For ease of exposition, we consider normalized data rate, i.e., $r(j) = \frac{\hat{r}(j)}{F}$, throughout this paper.


\begin{figure}
    \centering 
    \includegraphics[width=0.8\columnwidth,keepaspectratio]{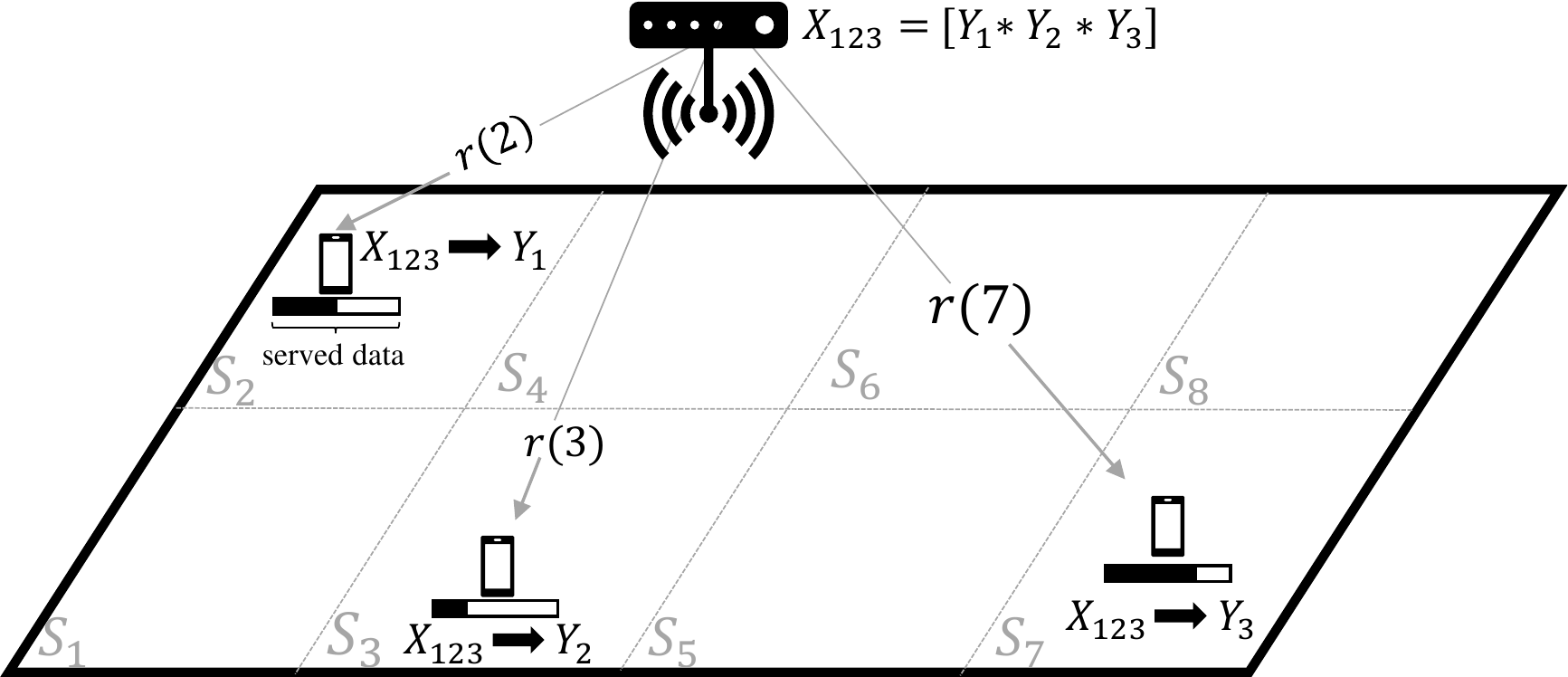}
    \caption{An application environment with $K=3$ users, split into $S=8$ states. $r(j)$ is the state-specific achievable rate and $r(3) > r(2) > r(7)$. $X_{123}$ is the multicast message, and $Y_i$ represents the data part intended for user $i$. The black bar below each user indicates how much of the requested data is cached, and $*$ denotes the variable-rate NCM operation.}
    \label{fig:system-model}
\end{figure}

\section{Location-Dependent Coded Caching}
\label{section:the_new_scheme}
Similar to other centralized coded caching schemes, our new location-dependent scheme also works in two distinct phases, 1) cache placement and 2) content delivery. 

\subsection{Cache Placement}
During the placement phase, the users' cache memories are filled up with useful data to minimize the duration of the content delivery phase. However, the content placement phase is carried out without any prior knowledge of users' spatial distribution in the delivery phase. As defined in Section~\ref{section:system_model}, we assume the application environment is split into $S$ states, and a location-dependent multimedia file of size one normalized data unit is required to reconstruct the static FoV of each state.

Before proceeding with the data placement, we first use a \textit{memory allocation} process to determine the dedicated amount of cache memory for storing (parts of) $W(j)$ at each user. Since there is no prior knowledge about the users' spatial realizations in the delivery phase, we minimize the maximum delivery time for a single user assuming uniform access probability for all the states. Let us use $m(j)$ to denote the normalized cache size at each user allocated to storing (parts of) $W(j)$. Since the size of $W(j)$ is normalized to one data unit, a user in state $j$ needs to receive $1-m(j)$ data units over the wireless link to reconstruct the FoV of state $j$. As a result, data delivery to state $j$ needs $T(j) = \frac{1-m(j)}{r(j)}$ seconds. Hence, the memory allocation for minimizing the maximum delivery time is found by solving the following linear program (LP):\footnote{Note that the achievable delivery time $\gamma$ can be controlled through $F$ and $B$ to meet the application's real-time requirements. In this regard, the solution of~\eqref{cache-allocation} can determine the maximum achievable QoE of different states (i.e., the maximum value of $F$ such that the constraints are met), using bisection over $F$ (and/or $B$). }
\begin{equation}
\begin{aligned}
\label{cache-allocation}
 \text{LP} \ : \quad &\min_{m(j), \gamma \ge 0}  \gamma\\
&\textrm{s.t.} \quad  \frac{1-m(j)}{r(j)} \leq \gamma, \ \forall j \in \CS, \\
& \sum_{j \in \CS} m(j) = M.
\end{aligned}
\end{equation}
%

After the memory allocation process, we store data in the cache memories of the users following the same method proposed in~\cite{maddah2014fundamental}. In this regard, for every $j\in\CS$ we split $W(j)$ into~$\binom{K}{t(j)}$ sub-files denoted by $W_{\CV(j)}(j)$, where $t(j) = K m(j)$ and $\CV(j)$ can be any subset of the user set $\CK$ with $|\CV(j)| = t(j)$.\footnote{Here we assume for every $j\in\CS$, $m(j)>0$ and $t(j)$ is an integer. In the next sections, it is briefly discussed what happens if these constraints are not met. However, a detailed discussion is left for the extended version.} Then, at the cache memory of user $i \in \CK$, we store $W_{\CV(j)}(j)$ for every state $j \in \CS$ and set $\CV(j) \ni i$. The cache placement procedure is outlined in Algorithm~\ref{Alg:placement}. 

\begin{exmp}
\label{exmp:placement}
Consider an application with $K=4$ users, where the environment is split into $S=5$ states and for each state, the required data size is $F=400$ Megabytes. Each user has a cache size of $900$ Megabytes, and hence, the normalized cache size is $M = 2.25$ data units. The spatial distribution of the achievable rate (assuming $B=F$) and its resulting memory allocation are as shown in Figure~\ref{fig:rate&cache_distribution}. It can be easily verified that $t(1) = t(5) = 1$, $t(2) = t(4) = 2$, and $t(3)=3$. As a result, $W(1)$, $W(3)$ and $W(5)$ should be split into four sub-files, while $W(2)$ and $W(4)$ are split into six sub-files. The resulting cache placement is visualized in Figure~\ref{fig:cache pool}.
\begin{figure}[t]
    \centering
    \resizebox{0.6\columnwidth}{!}{%

    \begin{tikzpicture}

    \begin{axis}
    [
    axis lines = left,
    xlabel near ticks,
    xlabel = \smaller {States},
    ylabel = \smaller {Rate [bits/Hz/s] \ref{rate}},
    ylabel near ticks,
    ymin = 0, 
    legend pos = north west,
    ticklabel style={font=\smaller},
    grid=both,
    major grid style={line width=.2pt,draw=gray!30},
    ]
    
    \addplot
    [mark = square, black]
    table[y=rates,x=States]{Figures/Data/Fig23.tex};
    \label{rate}
    \end{axis}
    
    \begin{axis}
    [
    axis lines = right,
    axis x line*=bottom,
    xtick={},
    xticklabels={},
    axis y line*=right,
    ylabel = \smaller {Memory \%  \ref{memory}},
    ylabel near ticks,
    ymin = 0, ymax = 100,
    ticklabel style={font=\smaller},
    major grid style={line width=.2pt,draw=gray!30},
    ]
    
    \addplot
    [mark = x, gray]
    table[y=Memory,x=States]{Figures/Data/Fig23.tex};
    \label{memory}
    \end{axis}

    \end{tikzpicture}
    }

    \caption{Location-specific rate and memory distributions for Example~\ref{exmp:placement}.
    }
    \label{fig:rate&cache_distribution}
\end{figure}
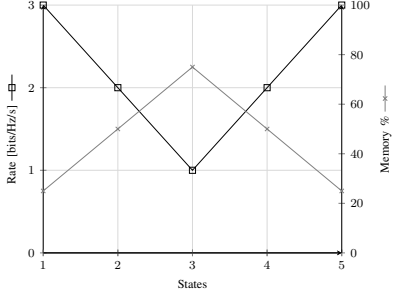
\begin{figure}
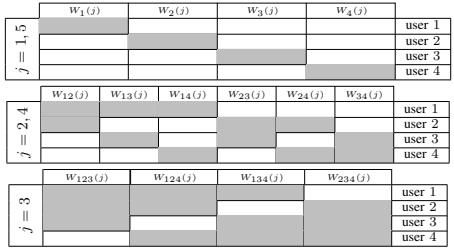

    \smaller
    \centering
    \resizebox{0.7\columnwidth}{!}{%
    \begin{tabular}{c}
        \begin{tabular}{|c|c|c|c|c|c|c|c|c|c|c|c|c|c|}
            \cline{2-13}
            \multicolumn{1}{c|}{} & \multicolumn{3}{c|}{\tiny$\,\qquad W_1(j) \qquad$} & \multicolumn{3}{c|}{\tiny$\,\qquad W_2(j) \qquad$} & \multicolumn{3}{c|}{\tiny$\,\qquad W_3(j) \qquad$} & \multicolumn{3}{c|}{\tiny$\,\qquad W_4(j) \qquad$} \\
            \hline
            \multirow{4}{*}{\rotatebox[origin=c]{90}{\small $j=1,5$}} & \multicolumn{3}{c|}{\cellcolor{gray!50}} & 
            \multicolumn{3}{c|}{} & 
            \multicolumn{3}{c|}{} & 
            \multicolumn{3}{c|}{} & user $1$ \\
            \cline{2-14}
            & \multicolumn{3}{c|}{} & 
            \multicolumn{3}{c|}{\cellcolor{gray!50}} & 
            \multicolumn{3}{c|}{} & 
            \multicolumn{3}{c|}{} & user $2$ \\
            \cline{2-14}
            & \multicolumn{3}{c|}{} & 
            \multicolumn{3}{c|}{} & 
            \multicolumn{3}{c|}{\cellcolor{gray!50}} & 
            \multicolumn{3}{c|}{} & user $3$ \\
            \cline{2-14}
            & \multicolumn{3}{c|}{} & 
            \multicolumn{3}{c|}{} & 
            \multicolumn{3}{c|}{} & 
            \multicolumn{3}{c|}{\cellcolor{gray!50}} & user $4$ \\
            \hline
            \multicolumn{14}{c}{}\\[-0.7em]
        \end{tabular} \\
        \begin{tabular}{|c|c|c|c|c|c|c|c|c|c|c|c|c|c|}
            \cline{2-13}
            \multicolumn{1}{c|}{} & \multicolumn{2}{c|}{\tiny$W_{12}(j)$} & \multicolumn{2}{c|}{\tiny$W_{13}(j)$} & \multicolumn{2}{c|}{\tiny$W_{14}(j)$} & \multicolumn{2}{c|}{\tiny$W_{23}(j)$} & \multicolumn{2}{c|}{\tiny$W_{24}(j)$} & \multicolumn{2}{c|}{\tiny$W_{34}(j)$}  \\
            \hline
            \multirow{4}{*}{\rotatebox[origin=c]{90}{\small $j=2,4$}} & \multicolumn{2}{c|}{\cellcolor{gray!50}} & 
            \multicolumn{2}{c|}{\cellcolor{gray!50}} & 
            \multicolumn{2}{c|}{\cellcolor{gray!50}} & 
            \multicolumn{2}{c|}{} & \multicolumn{2}{c|}{} & \multicolumn{2}{c|}{} & user $1$ \\
            \cline{2-14}
            & \multicolumn{2}{c|}{\cellcolor{gray!50}} & 
            \multicolumn{2}{c|}{} & 
            \multicolumn{2}{c|}{} & 
            \multicolumn{2}{c|}{\cellcolor{gray!50}} & \multicolumn{2}{c|}{\cellcolor{gray!50}} & \multicolumn{2}{c|}{} & user $2$ \\
            \cline{2-14}
            & \multicolumn{2}{c|}{} & 
            \multicolumn{2}{c|}{\cellcolor{gray!50}} & 
            \multicolumn{2}{c|}{} & 
            \multicolumn{2}{c|}{\cellcolor{gray!50}} & \multicolumn{2}{c|}{} & \multicolumn{2}{c|}{\cellcolor{gray!50}} & user $3$ \\
            \cline{2-14}
            & \multicolumn{2}{c|}{} & 
            \multicolumn{2}{c|}{} & 
            \multicolumn{2}{c|}{\cellcolor{gray!50}} & 
            \multicolumn{2}{c|}{} & \multicolumn{2}{c|}{\cellcolor{gray!50}} & \multicolumn{2}{c|}{\cellcolor{gray!50}} & user $4$ \\
            \hline
            \multicolumn{14}{c}{}\\[-0.7em]
        \end{tabular}  \\
        \begin{tabular}{|c|c|c|c|c|c|c|c|c|c|c|c|c|c|}
            \cline{2-13}
            \multicolumn{1}{c|}{} & \multicolumn{3}{c|}{\tiny$\;\quad W_{123}(j) \,\quad$} & \multicolumn{3}{|c}{\tiny$\;\quad W_{124}(j) \,\quad$} & \multicolumn{3}{|c}{\tiny$\;\quad W_{134}(j) \,\quad$} & \multicolumn{3}{|c|}{\tiny$\;\quad W_{234}(j) \,\quad$} \\
            \hline
            \multirow{4}{*}{\rotatebox[origin=c]{90}{\small $j=3$}} & \multicolumn{3}{c|}{\cellcolor{gray!50}} & 
            \multicolumn{3}{c|}{\cellcolor{gray!50}} & 
            \multicolumn{3}{c|}{\cellcolor{gray!50}} & 
            \multicolumn{3}{c|}{} & user $1$ \\
            \cline{2-14}
            & \multicolumn{3}{c|}{\cellcolor{gray!50}} & 
            \multicolumn{3}{c|}{\cellcolor{gray!50}} & 
            \multicolumn{3}{c|}{} & 
            \multicolumn{3}{c|}{\cellcolor{gray!50}} & user $2$ \\
            \cline{2-14}
            & \multicolumn{3}{c|}{\cellcolor{gray!50}} & 
            \multicolumn{3}{c|}{} & 
            \multicolumn{3}{c|}{\cellcolor{gray!50}} & 
            \multicolumn{3}{c|}{\cellcolor{gray!50}} & user $3$ \\
            \cline{2-14}
            & \multicolumn{3}{c|}{} & 
            \multicolumn{3}{c|}{\cellcolor{gray!50}} & 
            \multicolumn{3}{c|}{\cellcolor{gray!50}} & 
            \multicolumn{3}{c|}{\cellcolor{gray!50}} & user $4$ \\
            \hline
        \end{tabular}
    \end{tabular}%
    }
    \caption{Cache placement visualization for Example~\ref{exmp:placement}.}
    \label{fig:cache pool}
\end{figure}

\end{exmp}

As $W(j)$ is split into $\binom{K}{t(j)}$ sub-files and $\binom{K-1}{t(j)-1}$ sub-files are stored in the cache memory of each user, the total memory size dedicated to $W(j)$ at each user is
\begin{equation}
    \frac{\binom{K-1}{t(j)-1}}{\binom{K}{t(j)}} = \frac{t(j)}{K} = m(j) \; ,
\end{equation} 
and hence, the proposed algorithm satisfies the cache size constraints.
Note that, in comparison with~\cite{maddah2014fundamental}, here the required files in each state are considered as a separate library, and the cache placement algorithm in~\cite{maddah2014fundamental} is performed for each state independent of the others. Also, different from the existing works, here, files of different locations have distinct $t(j)$ values, which should be carefully considered in the delivery phase. For notational simplicity, throughout the paper we ignore the brackets and separators while explicitly writing $\CV(j)$, i.e., $W_{i \, k}(j) \equiv W_{\{i,k\}}(j)$.

\begin{algorithm}[t]
\small
	\caption{Location-based cache placement}
	\begin{algorithmic}[1]
		\Procedure{CACHE\_PLACEMENT}{}
		\State $\{m(j) \} =$ The result of the LP problem in~\eqref{cache-allocation}
		
		\ForAll{$j \in \CS$} 
		
		\State $t(j) = K \times m(j)$ 
		
		\State $W(j) \rightarrow \{W_{\CV(j)}(j) \; | \; \CV(j) \subseteq \CK, |\CV(j)| = t(j)\}$
		
		\ForAll{$\CV(j)$}
		\ForAll{$i \in \CK$}
		\If{$i \in \CV(j)$}
		    \State Put $W_{\CV(j)}(j)$ in the cache of user $i$
		\EndIf
		\EndFor
		\EndFor
		\EndFor
		\EndProcedure 
	\end{algorithmic}
	\label{Alg:placement}
\end{algorithm}

\subsection{Content Delivery}
At the beginning of the delivery phase, every user $i \in \CK$ reveals its requested file $W_i \equiv W(s_i).$\footnote{Note that we have used $W(j)$ to represent the file required for reconstructing the FoV of the state $j$, and $W_i$ to denote the file requested by user~$i$. The same convention is used for all notations in the text.} Note that $W_i$ depends on the state $s_i$ where user $i$ is located. The server then builds and transmits several \textit{nested} codewords, such that after receiving the codewords, all users can reconstruct their requested files. From the system model, user $i$ requires a total amount of one normalized data unit to reconstruct $W_i$. However, a subset of this data, with size $m_i \equiv m(s_i)$ data units, is available in the cache of user $i$. Since each user might be in a different state with distinct $t(j)$, the conventional delivery scheme of~\cite{maddah2014fundamental} is no longer applicable, and a new delivery mechanism is required to achieve a proper multicasting gain.

\begin{algorithm}[t]
\small
    \caption{NCM-based Content Delivery}
	\begin{algorithmic}[1]
		\Procedure{DELIVERY}{}
		
		\State $\hat{t} = \min_{i \in \CK} t_i$
		\ForAll{$\CU \subseteq \CK : |\CU| = \hat{t}+1$}
		    \State $X_{\CU} \gets 0$
		    \ForAll{$i \in \CU$}
		        \State $\alpha_i \gets \binom{t_i}{\hat{t}}$, $Y_{\CU,i} \gets 0$, $\CU_{-i} \gets \CU \backslash \{i \}$
		        \ForAll{$\CV_i \subseteq \CK : |\CV_i| = t_i+1$}
		            \If{$\CU_{-i} \subseteq \CV_i$, $i \not\in \CV_i$}
		                \State $W_{\CV_i,i}^q \gets$ \textsc{Chunk}($W_{\CV_i,i}, \alpha_i$)
		                \State $Y_{\CU,i} \gets$ \textsc{Concat} $ (Y_{\CU,i}, W_{\CV_i,i}^q)$ 
		            \EndIf
		        \EndFor
		        \State $X_{\CU} \gets$ \textsc{Nest} $(X_{\CU}, Y_{\CU,i}, r_i)$
		    \EndFor
		    \State Transmit $X_{\CU}$
		\EndFor
		\EndProcedure
	\end{algorithmic}
	\label{Alg:Delivery}
\end{algorithm}

The new delivery algorithm for the proposed setup is outlined in Algorithm~\ref{Alg:Delivery}. During the delivery phase, the server transmits a nested codeword $X_{\CU}$ for every subset of users $\CU$ with $|\CU|= \hat{t} +1$, where $\hat{t}$ is the \textit{common cache ratio} defined as $\hat{t} = \min_{i \in \CK} t_i$, and $t_i \equiv t(s_i)$. From the placement phase, we recall that the file $W_i$ intended for user $i \in \CU$ is split into sub-files $W_{\CV(s_i),i}$.
Let us define $\CU_{-i} \equiv \CU \setminus \{i\}$ and consider a user $k \in \CU_{-i}$. All the sub-files $W_{\CV(s_i),i}$ for which $k \in \CV(s_i)$ are already available in the cache memory of user $k$. Thus, by transmitting (part of) every sub-file $W_{\CV(s_i),i}$ for which $\CU_{-i} \subseteq \CV(s_i)$ and $i \not\in \CV(s_i)$, a portion of $W_i$ is delivered to user $i$ without causing any interference at other users $k \in \CU_{-i}$.

\begin{exmp}
\label{exmp:interference-free-delivery}
Consider the network in Example~\ref{exmp:placement}, for which the cache placement is visualized in Figure~\ref{fig:cache pool}. Assume in a specific time slot, $s_1 = 1$, $s_2 =2$, $s_3 =4$, $s_4 = 5$. Denoting the set of requested sub-files for user $i$ with $\CT_i$ and assuming $A \equiv W(1)$, $B \equiv W(2)$, $C \equiv W(4)$, $D \equiv W(5)$, we have
\begin{equation}
    \begin{aligned}
    \CT_1 &= \{ A_{2}, A_{3}, A_{4} \}, \qquad
    \CT_2 = \{ B_{13}, B_{14}, B_{34} \}, \\
    \CT_3 &= \{ C_{12}, C_{14}, C_{24} \}, \quad
    \CT_4 = \{ D_{1}, D_{2}, D_{3} \}. \\
    \end{aligned}
\end{equation}
Note that, the size of the sub-files of $A,B,C,D$ are $\frac{1}{4}, \frac{1}{6}, \frac{1}{6}, \frac{1}{4}$ data units, respectively. The common cache ratio is $\hat{t} = 1$, and hence, during each transmission we deliver data to $\hat{t}+1 = 2$ users.
Let us consider the case $\CU = \{1,2\}$. The codeword $X_{12}$ is supposed to deliver a portion of the requested data to users $1$ and $2$ interference-free. This is done by including (parts of) $A_{2}$, $B_{13}$, $B_{14}$ in $X_{12}$. As can be seen from Figure~\ref{fig:cache pool}, user $1$ has $B_{13}$ and $B_{14}$ in its cache memory, and so it can remove them from its received data and decode (part of) $A_{2} \in \CT_1$ interference-free. Similarly, user $2$ can remove $A_{2}$ using its cache contents and decode (parts of) $B_{13},B_{14} \in \CT_2$ interference-free. The coding procedure is explained in more details in Example~\ref{exmp:coding-procedure}.
\end{exmp}

As illustrated in Example~\ref{exmp:interference-free-delivery}, 
due to the proposed location-dependent cache placement, the value of $t_i$ (and hence, the length of the transmitted sub-files) might be different for various users. Moreover, we may find more than one sub-file to be transmitted to user $i \in \CU$. We address these issues by introducing a normalizing factor $\alpha_i = \binom{t_i}{\hat{t}}$ and two auxiliary functions $\textsc{Chunk}$ and $\textsc{Concat}$. Whenever $t_i > \hat{t}$ for user $i \in \CU$, we first split every sub-file intended for user $i$ into $\alpha_i$ smaller \textit{chunks} (denoted by $W_{\CV(s_i),i}^q$ in Algorithm~\ref{Alg:Delivery}). Then, we \textit{concatenate} a selection of these chunks to create the part of $X_{\CU}$ intended for user $i$, represented by $Y_{\CU,i}$. The function $\textsc{Chunk}$ ensures none of the chunks of a sub-file is sent twice, and the function $\textsc{Concat}$ creates a bit-wise concatenation of the given chunks. The final codeword $X_{\CU}$ is created by nesting $Y_{\CU,i}$ for every user $i \in \CU$. This is shown by the auxiliary function $\textsc{Nest}$ in Algorithm~\ref{Alg:Delivery}. Note that due to the nesting, every $Y_{\CU,i}$ can be transmitted with rate $r_i \equiv r(s_i)$ (c.f.~\cite{tang2011full}). 

\begin{exmp}
\label{exmp:coding-procedure}
Following Example~\ref{exmp:interference-free-delivery}, let us review how the codeword $X_{12}$ is built. This codeword includes (parts of) $A_{2}$, $B_{13}$, $B_{14}$, intended for users $1$ and $2$. However, the size of $A_2$ is $\frac{1}{4}$, while $B_{13},B_{14}$ are both $\frac{1}{6}$ data units. To solve this issue, we use normalizing factors $\alpha_1 = 1$ and $\alpha_2 = 2$, and build
\begin{equation}
    X_{12} = A_2 * \textsc{Concat} (B_{13}^{1}, B_{14}^{1}) \; ,
\end{equation}
where the operator ($*$) denotes the nesting operation and superscripts are used to differentiate various chunks of a sub-file. The nesting operation in $X_{12}$ is performed such that $A_2$ and $\textsc{Concat} (B_{13}^{1}, B_{14}^{1})$ are delivered with rates $r_1 = 3$ and $r_2 = 2$ data units per second, respectively. As a result, this codeword is delivered in $\max(\frac{1}{3} \times \frac{1}{4}, \frac{1}{2} \times \frac{2}{12}) = \frac{1}{12}$ seconds. 
Similarly, as $\alpha_3 = 2$ and $\alpha_4 = 1$, we have
\begin{equation}
    \begin{aligned}
    X_{13} &= A_3 * \textsc{Concat}(C_{12}^{1},C_{14}^{1}) \; , \qquad X_{14} = A_4 * D_1 \; , \\
    X_{23} &= \textsc{Concat}(B_{13}^{2} \; B_{34}^{1}) * \textsc{Concat}(C_{12}^{2} \; C_{24}^{1}) \; , \\
    X_{24} &= \textsc{Concat}(B_{14}^{2} \; B_{34}^{2}) * D_2 \; , \\
    X_{34} &= \textsc{Concat}(C_{14}^{4} \; C_{24}^{4}) * D_3 \; , \\
    \end{aligned}
\end{equation}
where each transmission requires $\frac{1}{12}$ seconds. Thus, the total delivery time is $\frac{6}{12}$ seconds. In comparison, it can be seen that the delivery time would be doubled without multicasting.
\end{exmp}

\begin{thm}
Using the proposed cache placement and content delivery algorithms, every user receives its requested data.
\end{thm}
\begin{proof}
A user $i \in \CK$ needs a total amount of one (normalized) data unit to reconstruct its requested file $W_i$. If the user is in state $s_i$, $m(s_i)$ data units are available in its cache memory, and hence it needs to receive $1-m(s_i) = 1 - \frac{t_i}{K}$ data units from the server. As detailed in the delivery algorithm, user $i$ receives parts of its requested data in $\binom{K-1}{\hat{t}}$ transmissions. Consider one such transmission, where the codeword $X_{\CU}$ is transmitted to the users in $\CU \ni i$. During this transmission, user $i$ receives $Y_{\CU,i}$, which is built through concatenation of $\binom{K-\hat{t}-1}{t_i - \hat{t}}$ chunks, each with size
    $\nicefrac{1}{\alpha_i \binom{K}{t_i}} = \nicefrac{1}{\binom{K}{t_i} \binom{t_i}{\hat{t}}}$
data units. As a result, the total data size delivered to user $i$ from the server is
\begin{equation}
    \frac{ \binom{K-1}{\hat{t}} \binom{K-\hat{t}-1}{t_i - \hat{t}} }{ \binom{K}{t_i} \binom{t_i}{\hat{t}} } = \frac{K-t_i}{K} = 1 - \frac{t_i}{K} \; ,
\end{equation}
and the proof is complete.
\end{proof}

\section{Performance Analysis}
\label{section:performance}
We discuss the performance of the proposed scheme assuming that for every state $j \in \CS$, $t(j)$ is a positive integer. Then we briefly review what happens if this assumption is relaxed.

\begin{figure*}[ht]
\begin{minipage}{0.32\linewidth}
    \centering
    \resizebox{\columnwidth}{!}{%

    \begin{tikzpicture}

    \begin{axis}
    [
    axis lines = left,
    xlabel near ticks,
    xlabel = \smaller {The number of users ($K$)},
    ylabel = \smaller {Delivery time [ms]},
    ylabel near ticks,
    ymin = 0, 
    legend pos = north west,
    ticklabel style={font=\smaller},
    grid=both,
    major grid style={line width=.2pt,draw=gray!30},
    ]
    
    \addplot
    [mark = square, black]
    table[y=Tu,x=K]{Figures/Data/Fig4.tex};
    \addlegendentry{\tiny $T_u$}
    
    \addplot
    [mark = x, black!60]
    table[y=Tx,x=K]{Figures/Data/Fig4.tex};
    \addlegendentry{\tiny $T_x$}
    
    \addplot
    [dashed, black!40]
    table[y=Tm,x=K]{Figures/Data/Fig4.tex};
    \addlegendentry{\tiny $T_m$}
    
    \end{axis}

    \end{tikzpicture}
    }

    \caption{Average delivery time vs the user count~($K$), $M = \frac{3}{4}S$.
    }
    \label{fig:timegain_Kx}
\label{fig:figure1}
\end{minipage}%
\hfill
\begin{minipage}{0.3\linewidth}
    \centering
    \resizebox{\columnwidth}{!}{%

    \begin{tikzpicture}

    \begin{axis}
    [
    axis lines = left,
    xlabel near ticks,
    xlabel = \smaller {M [Data Units]},
    ylabel = \smaller {Performance Ratio},
    ylabel near ticks,
    ymin = 0, 
    legend pos = north west,
    ticklabel style={font=\smaller},
    grid=both,
    major grid style={line width=.2pt,draw=gray!30},
    ]
    
    \addplot
    [mark = x, black]
    table[y=Tu-Tx-S121,x=M]{Figures/Data/Fig56.tex};
    \addlegendentry{\tiny $\frac{T_u}{T_x}$}
    
    \addplot
    [mark = square, black!60]
    table[y=Tu-Tm-S121,x=M]{Figures/Data/Fig56.tex};
    \addlegendentry{\tiny $\frac{T_u}{T_m}$}
    
    \end{axis}

    \end{tikzpicture}
    }

    \caption{Relative gain in delivery time vs the cache size~($M$), $K = 10$, $S = 121$.
    }
    \label{fig:time_Mxs121}
\label{fig:figure2}
\end{minipage}%
\hfill
\begin{minipage}{0.3\linewidth}
    \centering
    \resizebox{\columnwidth}{!}{%

    \begin{tikzpicture}

    \begin{axis}
    [
    axis lines = left,
    xlabel near ticks,
    xlabel = \smaller {M [Data Units]},
    ylabel = \smaller {Performance Ratio},
    ylabel near ticks,
    ymin = 0, 
    legend pos = north west,
    ticklabel style={font=\smaller},
    grid=both,
    major grid style={line width=.2pt,draw=gray!30},
    ]
    
    \addplot
    [mark = x, black]
    table[y=Tu-Tx-S441,x=M]{Figures/Data/Fig56.tex};
    \addlegendentry{\tiny $\frac{T_u}{T_x}$}
    
    \addplot
    [mark = square, black!60]
    table[y=Tu-Tm-S441,x=M]{Figures/Data/Fig56.tex};
    \addlegendentry{\tiny $\frac{T_u}{T_m}$}
    
    \end{axis}

    \end{tikzpicture}
    }

    \caption{Relative gain in delivery time vs the cache size~($M$), $K = 10$, $S = 441$.
    }
    \label{fig:timegain_Mx441}
\label{fig:figure3}
\end{minipage}
\vspace{-15pt}
\end{figure*}

\begin{lem} \label{lemma:eqality}
For every two states $j,j' \in \CS$, the result of the memory allocation problem in~\eqref{cache-allocation} satisfies
\begin{equation} \label{eq:optimalms}
    \frac{1-m(j)}{r(j)} = \frac{1-m(j')}{r(j')}\;, \quad  \frac{1-m(j)}{r(j)} = \frac{S-M}{\sum_{j' \in \CS} r(j')} \; .
\end{equation}
\end{lem}
\begin{proof}
The first equality can be simply proved by contradiction.
Using this equality, for the second one we can write
\begin{equation}
    \frac{1-m(j)}{r(j)} \sum_{j'\in\CS} r(j') = \sum_{j'\in\CS} 1-m(j') = S-M  \; ,
\end{equation}
and the proof is complete.
\vspace{-5pt}
\end{proof}

\begin{thm}\label{theorem:proposed transmission time}
Total delivery time of the proposed scheme, denoted by $T_m$, is calculated as
\begin{equation}\label{eq:proposed transmission time}
    T_m = \frac{K}{\hat{t} + 1} \frac{S-M}{\sum_{j \in \CS}r(j)} \; .
\end{equation}
\end{thm}
\begin{proof}
A user $i \in \CK$ receives $1-m(s_i)$ data units from the server through $n = \binom{K-1}{\hat{t}}$ transmissions. However, as the same procedure is followed to build $Y_{\CU,i}$ for every $\CU \ni i$, the data size delivered to user $i$ is the same for all these transmissions. Moreover, this user is always served with the rate $r_i$, and hence, it needs $T_i = \frac{1}{n} \frac{1-m(s_i)}{r(i)}$ seconds to receive every data part. Using Lemma~\ref{lemma:eqality}, we can rewrite $T_i$ as 
    $T_i = \frac{1}{n} \frac{S-M}{\sum_{j \in \CS} r(j)}$,
which means the time required for all the transmissions is the same and independent of the user selection. As in the delivery phase there exist a total number of $\binom{K}{\hat{t}+1}$ transmissions, the total delivery time is
\begin{equation}
    T_m = \frac{\binom{K}{\hat{t}+1}}{\binom{K-1}{\hat{t}}} \frac{S-M}{\sum_{j \in \CS}r(j)} = \frac{K}{\hat{t} + 1}\frac{S-M}{\sum_{j \in \CS}r(j)} \; ,
\end{equation}
and the proof is complete.
\end{proof}

\begin{thm}
\label{thm:time_gain}
Compared with unicasting, the proposed scheme enables a reduction in the delivery time by a factor of $\hat{t}+1$.
\end{thm}
\begin{proof}
In case of unicasting, a user $i$ gets $1-m(s_i)$ data units with the rate $r_i$, and hence, the total delivery time is equal to
    $T_u = \sum_{i \in \CK} \frac{1-m(s_i)}{r_i}$.
However, using Lemma~\ref{lemma:eqality} we have
\begin{equation}
    T_u = \sum_{i \in \CK} \frac{S-M}{\sum_{j\in\CS}r(j)} = K\frac{S-M}{\sum_{j\in\CS}r(j)} \; ,
\end{equation}
which is $\hat{t}+1$ times larger than the $T_m$ value in~\eqref{eq:proposed transmission time}.
\end{proof}

\begin{cor} 
The proposed scheme becomes more efficient if $\hat{t}$ is larger.
This happens, for example, when more users are located in states with poor channel conditions.
\end{cor}

\begin{thm}
Under the assumption of uncoded cache placement, the proposed delivery time in~\eqref{eq:proposed transmission time} is within a bounded gap from any optimal scheme and satisfies
\begin{equation} \label{bound:performance_gap_mn}
    \frac{T_m}{T^{*}} \le \frac{\bar{r}}{(\hat{m}+\frac{1}{K}) \sum_{j \in \CS}r(j)}\left(\frac{S}{K}+M\right) \; ,
\end{equation}
where $\bar{r} = \max_{j \in \CS} r(j)$, $\hat{m} = \min_{j \in \CS} m(j)$, and $T^*$ is the optimal delivery time with uncoded cache placement. 
\end{thm}

\begin{proof}
For proof, we use optimality results provided for the original scheme of~\cite{maddah2014fundamental}. Note that if $r(j)$ is the same for every state $j \in \CS$, our scheme becomes equivalent to~\cite{maddah2014fundamental}, with library size $S$ and coded caching gain $\frac{KM}{S}$. Moreover, in~\cite{yu2017exact}, it is shown that under the assumption of uncoded cache placement, the original coded caching scheme in~\cite{maddah2014fundamental} is exactly optimal. So, we can provide an upper bound for the optimal delivery time in our setup by simply assuming that the achievable rate at every state is the same and equal to $\bar{r}$. Using the delivery time expressions in~\cite{maddah2014fundamental}, we can write
\begin{equation}
    \frac{K(S-M)}{S+KM} \frac{1}{\bar{r}} \le T^* \; ,
\end{equation}
which, together with~\eqref{eq:proposed transmission time}, results in
\begin{equation}
\label{eq:bound_int_relation}
    \frac{T_m}{T^*} \le \frac{K(S-M)}{(\hat{t}+1)\sum_{j \in \CS} r(j)} \frac{S+KM}{K(S-M)} \bar{r} \; .
\end{equation}
However, as we have
\begin{equation}
    \hat{t} = \min_{i \in \CK} t(s_i) \ge \min_{j \in \CS} t(j) = K \min_{j \in \CS} m(j) = K \hat{m} \; ,
\end{equation}
we can simply substitute $\hat{t}$ with $K \hat{m}$ in~\eqref{eq:bound_int_relation}, to reach the upper bound proposed in~\eqref{bound:performance_gap_mn}. Note that the right hand side expression in~\eqref{bound:performance_gap_mn} approaches to
\begin{equation}
    \frac{\bar{r}M}{\hat{m}\sum_{j \in \CS} r(j)}
\end{equation}
as $K \rightarrow \infty$, and is bounded by 
\begin{equation}
    \frac{\bar{r}}{K(\hat{m}+\frac{1}{K})\min_{j\in\CS}r(j)}
\end{equation}
as $S \rightarrow \infty$. Hence, in both cases, the proposed delivery time in our scheme is within a bounded gap from any optimal scheme, and the proof is complete.

\end{proof}

It can be shown that, if all $t(j)$ values are not integers, all the results provided in this section still hold. However, if $m(j)=0$ for some $j \in \CS$, the reduction factor in delivery time, proposed in Theorem~\ref{thm:time_gain}, is no longer valid. This is because, in such a case, it is not always possible to deliver all the required data through multicasting (some missing parts should be delivered by unicasting). 



\section{Simulation Results}
\label{section:sim_results}
\label{sec:simulation}
We use numerical simulations to evaluate the performance of the proposed location-dependent scheme. For simulations, we consider a $5m \times 5m$ square room where the transmitter is located in the middle of the room at a height of $3m$. The transmitter power is $30\mathrm{dB}$, the room area is split into $S=121$ states, and the users are uniformly distributed over the states. We compare the delivery time of the proposed scheme with unicasting and the baseline scheme of~\cite{tolli2017multi}. These delivery times are denoted by $T_m$, $T_u$ and $T_x$, respectively. For the scheme of~\cite{tolli2017multi}, $m(j) = \frac{M}{S}$ for every state $j \in \CS$, and the multicast rate for a user set $\CU$ is limited by $\min_{i \in \CU}r(s_i)$.



In Figure~\ref{fig:timegain_Kx}, we have compared the proposed location-dependent scheme with unicast transmission, for different values of the user count parameter $K$. Clearly, the new scheme not only reduces the delivery time, but also makes it (almost) independent of $K$. This independence is justified by Theorem~\ref{theorem:proposed transmission time}, which states that for integer $\hat{t}$, $T_m$ is proportional to $\frac{K}{\hat{t}+1}$. However, as $\hat{t}=\min_{u_i\in\CK} t_i$ and $t_i = K m(s_i)$, as $K$ becomes larger, the value of $T_m$ approaches a constant value.



In Figure~\ref{fig:time_Mxs121}, the same comparison is done for different values of the cache size parameter $M$. It can be seen that the relative transmission gain $\frac{T_u}{T_m}$ increases almost linearly when $M > 0.4$. This is because for $M < 0.4$, there exists at least one state $j\in\CS$ for which $t(j) < 1$, and hence, the efficiency of the proposed scheme is reduced as multicasting is not possible when some users are located in such states. On the other hand, it can be seen that the proposed scheme performs worse than the baseline scheme of~\cite{tolli2017multi} when $M$ is small.
This is because our proposed scheme sacrifices the global caching gain for a higher local caching gain. While this may not seem appealing, it results in more resilient performance if the ratio between the best and worst channel conditions is large. This is investigated in Figure~\ref{fig:timegain_Mx441}, where we have repeated the simulations for a larger room area of $10m \times 10m$, split into $S=441$ states. Clearly, in this case, the proposed location-dependent scheme outperforms the baseline scheme of~\cite{tolli2017multi} with a good margin, confirming its resilience to the existence of ill-conditioned states. 

\section{Conclusions and Future Work}
\label{section:conclusion}
In this paper, we proposed a centralized, location-dependent coded caching scheme, tailored for future immersive viewing applications. For the placement phase, we use a memory allocation process to allocate larger cache portions where the channel condition is poor, and during the delivery phase, we use nested code modulation to support different data rates within a single multicast transmission. The resulting scheme provides an (almost) uniform QoE throughout the application environment, and performs better than the state of the art in ill-conditioned scenarios where the ratio between the best and worst channel conditions is large.


The proposed location-dependent scheme can be extended in various directions. An extension to support multi-antenna communication setups is currently in progress. Other notable research opportunities include supporting multiple transmitters, incorporating side information on user movement patterns and state transition probabilities, and considering a more dynamic scenario where the users' cache content is updated as they move through the application environment.



\bibliographystyle{IEEEtran}
\bibliography{references}

\end{document}